\documentclass{amsart}

\usepackage{latexsym}
\usepackage{amsmath}
\usepackage{epsfig}
\usepackage{amssymb}
\usepackage{enumerate}

\newtheorem{theorem}{Theorem}[section]
\newtheorem{lemma}[theorem]{Lemma}

\newtheorem{corollary}[theorem]{Corollary}
\newtheorem{proposition}[theorem]{Proposition}

\newcommand{\T}{{\mathcal T}}

\newcommand{\cT}{{\mathcal T}}

\newcommand{\Var}{\mathop{\textrm{Var}}\nolimits}
\newcommand{\Cov}{\mathop{\textrm{Cov}}\nolimits}

\title[Distribution of phylogenetic diversity]{Distribution of phylogenetic diversity under random extinction}

\author{Be{\'a}ta Faller, Fabio Pardi and Mike Steel}

\thanks{We thank the NZ Marsden Fund (06-UOC-02) for
supporting this research.}

\address{European Bioinformatics Institute, Wellcome Trust Genome
Campus, Hixton, Cambridge, UK}

\email{pardi@ebi.ac.uk}

\address{Biomathematics Research Centre, Department of Mathematics and
  Statistics, University of Canterbury, Christchurch, New Zealand}

\email{fallerbeata@yahoo.com, m.steel@math.canterbury.ac.nz}

\subjclass{05C05; 92D15}

\keywords{biodiversity conservation, phylogenetic tree, central
limit theorem, field of bullets model of extinction}

\date{3 August 2007}

\begin{document}

\begin{abstract}
Phylogenetic diversity is a measure for describing how much of an
evolutionary tree is spanned by a subset of species.   If one
applies this to the (unknown) subset of current species that will
still be present at some future time, then this `future phylogenetic
diversity' provides a measure of the impact of various extinction
scenarios in biodiversity conservation. In this paper we study the
distribution of future phylogenetic diversity under a simple model
of extinction (a generalized `field of bullets' model). We show that
the distribution of future phylogenetic diversity converges to a
normal distribution as the number of species grows (under mild
conditions, which are necessary). We also describe an algorithm to
compute the distribution efficiently, provided the edge lengths are
integral, and briefly outline the significance of our findings for
biodiversity conservation.
\end{abstract}

\maketitle

Corresponding Author:

Mike Steel

Phone: +64-3-3667001, Ext. 7688

Fax: +64-3-3642587

Email: m.steel@math.canterbury.ac.nz

\newpage

\section{Introduction}

The current rapid rate of extinction of many diverse species has
focused attention on predicting the loss of future biodiversity.
There are numerous ways to measure the `biodiversity' of a group of
species, and one which recognises the evolutionary linkages between
taxa (for example, species) is phylogenetic diversity (\cite{fai92},
\cite{fai06}, \cite{moo}). Briefly, given a subset of taxa, the PD
(phylogenetic diversity) score of that subset is the sum of the
lengths of the edges of the evolutionary tree that connects this
subset (formal definitions are given shortly). Here the `length' of
an edge may refer to the amount of genetic change on that edge,
its temporal duration or perhaps other features (such as morphological
diversity).

Under the simplest models of speciation, each taxon has the same
probability of surviving until some future time, and the survival of
taxa are treated as independent events; this is a simple type of
`field of bullets' model (\cite{nee}, \cite{raup}, \cite{vaz}). This
model is quite restrictive (\cite{pur}) and a more realistic
extension allows each species to have its own survival probability
-- this is the model we study in this paper. Under this model, we
would like to be able to predict the PD score of the set of taxa
that survive. This `future PD' is a random variable with a
well-defined distribution, but to date, most attention has focused on
just its mean (that is, the expected PD score of the species that
survive). For example, the `Noah's Ark problem' (\cite{hart06, noah,
pardi}) attempts to maximize expected future PD by allocating
resources that increase the survival probabilities in a constrained
way. Clearly, one could consider other properties of the distribution
of future PD -- for example the probability (let us call it the
$PL_0$ value) that future PD is less than some critical lower limit
($L_0$). Given different conservation strategies, we may wish to
maximize expected PD or minimize the $PL_0$ value.   A natural
question is how are these two quantities related?

To address these sorts of questions, we need to know the full
distribution of future PD. In this paper, we show that for large
trees, future PD is (asymptotically) normally distributed. Given the
increasing trend in biology of constructing and analysing
phylogenetic trees that contain large numbers of species
($10^2-10^3$), we see this result as timely. Our work was also
motivated by the suggestive form of distributions obtained by
simulating future PD by sampling 12-leaf subtrees randomly from
64-leaf trees from Nee and May (\cite{nee}, see also \cite{vaz}). To
formally prove the normal limit law requires some care, as future PD
is not a sum of independent random variables (even though the
survival events for the taxa at the leaves are treated
independently); consequently, the usual central limit theory does not
immediately apply.

This limit law has some useful consequences for applications. For
example, it means that for a large tree, the $PL_0$ value can be
estimated by the area under a normal curve to the left of
$\frac{L_0-\mathbb{E}[PD]}{\sqrt{\Var[PD]}}$. In particular, we see
that the relation between the $PL_0$ value and expected future PD
($\mathbb{E}[PD]$) involves  scaling by the standard deviation of
future PD (so strategies that aim to maximize expected future PD may
not necessarily minimize the $PL_0$ value).

Our normal distribution result is asymptotic - that is, it holds for
large trees. However, it is also useful to have techniques for
calculating the exact PD distribution on any given tree. In
Section~\ref{algorithm}, we show how this may be achieved by a
polynomial time algorithm under the mild assumption that each edge
length is an integer multiple of some fixed length. In
Section~\ref{unroot}, we show how our results can be easily modified
to handle an `unrooted' form of PD that has also been considered in
the literature.

\subsection{Definitions and preliminaries}

Throughout this paper $X$ will denote a set of {\em taxa} (for
example, different species, different genera or populations of
the same species) and $X'$ will denote a subset of $X$.
A {\em rooted phylogenetic $X$--tree} is a rooted
tree in which (i) all edges are oriented away from the root, (ii) $X$ is
the set of leaves (vertices of the tree with no outgoing edges) and
(iii) every vertex except the leaves (and also possibly the root)
has at least two out-going edges (allowing the root to have just one
outgoing arc will be useful later). In systematic biology, these
trees are used to represent evolutionary development of the set $X$
of taxa from their common ancestor (the root of the tree), and the
orientation of the edges corresponds to temporal ordering. Given a
rooted phylogenetic $X$--tree $\T$, we let $E(\T)$ denote the set of
edges, and $E_P(\T)$ denote the set of pendant edges (edges
that are incident with a leaf).

Suppose we have a rooted phylogenetic $X$--tree $\T$ and a map
$\lambda$ that assigns a non-negative real-valued length $\lambda_e$
to each edge $e$ of $\T$. Given the pair $(\T, \lambda)$ and a
subset $X'$ of $X$, the {\em phylogenetic diversity} of $X'$, denoted
$PD_{(\T, \lambda)}(X')$ -- or, more briefly, $PD(X')$ -- is the sum
of the $\lambda_e$ values of all edges that lie on at least one path
between an element of $X'$ and the root of $\T$.
\begin{figure}[ht] \begin{center}
\resizebox{13cm}{!}{
\input{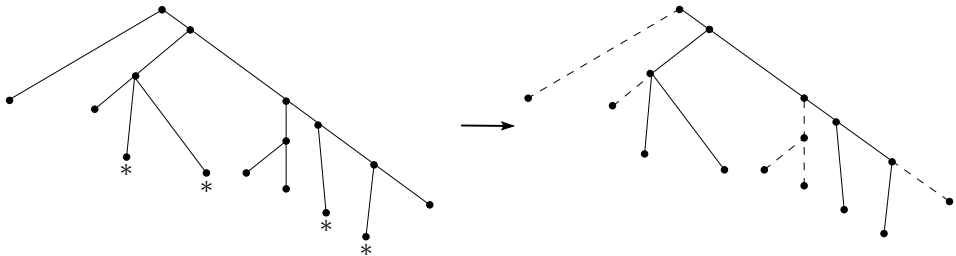}
} \caption{If only the taxa marked * in the tree on the left survive
then the future phylogenetic diversity is the sum of the lengths of
the solid edges in the tree on the right.}
\end{center}
\label{figure1}
\end{figure}
In the {\em (generalized) field of bullets model} (g-FOB), we have a
triple $(\T, \lambda, p)$ where $\T$ is a rooted phylogenetic
$X$--tree, $\lambda$ is an edge length assignment map, and $p$ is a
map that assigns to each leaf $i\in X$ a probability $p_i$.
Construct a random set $X'$ by assigning each element $i$ of $X$ to
$X'$ independently with probability $p_i$.  In biodiversity
conservation we regard $X'$ as the set of taxa that will still exist
(that is, not be extinct) at some time $t$ in the future;
accordingly, we call $p_i$ the {\em survival probability} of $i$.

Considering the random variable $\varphi =\varphi_\T =  PD_{(\T,
\lambda)}(X')$, which is the phylogenetic diversity of the random
subset $X'$ of $X$ (consisting of those taxa that `survive')
according to the process just described, we call $\varphi$  {\em
future phylogenetic diversity}. An example of this process is shown
in Fig. 1.

Note that in the g-FOB model, we can write
\begin{equation}
\label{fundy}
\varphi = \sum_e \lambda_eY_e,
\end{equation}
where $Y_e$ is the binary random variable which takes the value $1$
if $e$ lies on at least one path between an element of $X'$ and the
root of $\T$, and which is $0$ otherwise. Moreover,
\begin{equation}
\label{fundy2}
\mathbb{P}[Y_e = 1] = 1 - \prod_{i \in C(e)}(1-p_i),
\end{equation}
where $C(e)$ is the set of elements of $X$ that are separated from
the root of $\T$ by $e$. Consequently, if we let $$P_e: =
\mathbb{P}[Y_e=1] = 1-\prod_{i \in C(e)}(1-p_i),$$ then
\begin{equation}
\label{fundy3} \mathbb{E}[\varphi] =\sum_e\lambda_e P_e.
\end{equation}

Equation (\ref{fundy}) suggests that for large trees, $\varphi$ might
be normally distributed, as it will be sum of many random variables
(a normal distribution is also suggested by simulations described in
\cite{nee, vaz}). However, the random variables ($\lambda_eY_e$) are
not identically distributed and, more importantly, they are not
independent. Therefore a straightforward application of the (usual)
central limit theorem seems problematic. We show that under two mild
restrictions, a normal law can be established for large trees.
Moreover, neither of these two mild restrictions can be lifted (we
exhibit a counter-example to a normal law in both cases).

Since a normal distribution is determined once we know both its mean
and variance, it is useful to have equations for calculating both
these quantities.  Equation (\ref{fundy3}) provides a simple
expression for the mean, and we now present an  expression for the
variance that is also easy to compute. Given two distinct edges of
$\T$, we write $e <_{\T} f$ if the path from the root of
$\T$ to $f$ includes edge $e$ (or, equivalently,  $C(f) \subset C(e)$).

\begin{lemma}
\label{fundy4}
$$\Var[\varphi] = \sum_e \lambda_e^2P_e(1-P_e)  +
2\sum_{(e,f): e<_{\T}f} \lambda_e\lambda_f P_f(1-P_e).$$
\end{lemma}
\begin{proof}

>From Equation (\ref{fundy}) we have:
\[
\Var[\varphi]=\Cov[\varphi,\varphi]=\sum_{e,f}\lambda_e\lambda_f\Cov[Y_e,Y_f].
\]
The covariance of $Y_e$ and $Y_f$ is
$$\Cov[Y_e,Y_f]=\mathbb{E}[Y_e Y_f]-\mathbb{E}[Y_e]\mathbb{E}[Y_f]=\mathbb{P}[Y_e=1,Y_f=1]-\mathbb{P}[Y_e=1]\mathbb{P}[Y_f=1].$$
Now, we have the following cases:
\begin{itemize}
\item[(1)] \mbox{ } $e \neq f$ and neither $e<_{\T} f$ nor $f<_{\T} e$.  In this case, the subtree of $\T$ with root edge $e$ and the subtree of $\T$ with root edge $f$
do not have any leaves in common, and so $Y_e$ and $Y_f$ are independent. Thus, $\Cov[Y_e,Y_f]=0$.\\
\item [(2)] \mbox{ } $e<_{\T} f$. In this case, $C(f) \subset C(e)$ and so the survival of any taxon in $C(f)$ implies
the survival of a taxon in $C(e)$; that is, $Y_f=1$ implies $Y_e=1$
and we have
$\Cov[Y_e,Y_f]=\mathbb{P}[Y_f=1]-\mathbb{P}[Y_e=1]\mathbb{P}[Y_f=1] =P_f(1-P_e).$\\
\item[(3)]  \mbox{ } $f<_{\T} e$. This is analogous to case (2) (and, together with case (1), explains the factor of $2$ in the expression on the right-hand side of our formula for $\Var[\varphi]$). \\
\item[(4)] \mbox{ } $e=f$. This case gives $\Cov[Y_e,Y_f]=\mathbb{P}[Y_e=1](1-\mathbb{P}[Y_e=1]) =P_e(1-P_e)$
(and corresponds to the first term on the right hand-side of our formula for $\Var[\varphi]$).   \\
\end{itemize}
By considering these cases for $\Cov[Y_e,Y_f]$, we obtain the result claimed.
\end{proof}

A consequence of this lemma is the following lower bound on the
variance of future $PD$ which will be useful later.

\begin{corollary}
\label{lem1} Consider the g-FOB model on $(\T, \lambda, p)$.  Then,
$$
\Var[\varphi] \geq \sum_{e \in E_P(\T)}
\lambda_e^2P_e(1-P_e).$$
\end{corollary}
\begin{proof}
Notice that all the terms in the summation
expression for $\Var[\varphi]$ in Lemma~\ref{fundy4} are
non-negative, and so a lower bound on $\Var[\varphi]$ is obtained by
summing over those pairs $(e,f)$ for which $e=f$ is a pendant edge
of $\T$. This gives the claimed bound.
\end{proof}

\section{Asymptotic normality of future phylogenetic diversity under the g-FOB model}

Consider a sequence of such rooted phylogenetic trees:
\[
\T_1,\T_2,\ldots,\T_n,\ldots
\]
where $\T_n$ has a leaf label set $X= \{1,\ldots, n\}$. Furthermore, suppose
that for each tree we have an associated edge length function
$\lambda = \lambda^{(n)}$ and a survival probability function $p =
p^{(n)}$. For the sequence of g-FOB models $(\T_n, \lambda^{(n)},
p^{(n)})$, we impose the following conditions (where
$E_P(\T_n)$ is the set of pendant edges of $\T_n$):

\begin{itemize}
\item[({\bf C1})] For some $\epsilon>0$ and for each $n$, we have:
$$\epsilon \leq p_i^{(n)} \leq 1-\epsilon,$$ for
all $i \in \{1, \ldots, n\}$ except for at most $An^{\alpha}$ values of $i$, where $A, \alpha \geq 0$ are constants, with
$\alpha<\frac{1}{2}$.

\item[({\bf C2})] Let $L(n) = \max\{\lambda_e^{(n)}: e \in
E(\T_n)\}$. Then, for each $n$, we have: $$\sum_{e \in
E_P(\T_n)}\left(\lambda_e^{(n)}\right)^2 \geq B n^{\beta}L(n)^2,$$
for some constants $B>0, \beta > 2\alpha$.
\end{itemize}

\noindent{\bf Remarks concerning conditions (C1), (C2).}

Condition (C1) simply says that the survival of most taxa is
neither (arbitrarily close to) certain nor impossible. The term
$An^{\alpha}$ provides the flexibility to allow for some of the taxa
to have a survival probability that is very close to, or even equal
to, 0 or $1$.

Condition (C2)  says, roughly speaking, that the pendant edges are,
on average, not too short in relation to the longest edge in the
tree.  This is relevant for evolutionary biology, as it follows that
for trees generated by a constant speciation rate `pure birth' model
(see, for example, \cite{yule}) condition (C2) holds in expectation
(for any $\alpha \in (0,\frac{1}{2})$). A more formal statement of
this claim, and its proof, is given in the Appendix.

Note that if condition (C2) holds for a value $\beta>0$ then,
$\beta$ is at most $1$, since the terms in the summation expression
in (C2) are all at most $1$ and there are $O(n)$ of them.

\hfill $\Box$

Next, we state our main theorem, which describes the asymptotic
normality of future phylogenetic diversity $\varphi_n =
\varphi_{\T_n}$. Since phylogenetic trees often contain a large
number of taxa, the result allows one to approximate the
distribution of future phylogenetic diversity with a normal
distribution.

\begin{theorem}\label{tetelegy}
Under conditions (C1) and (C2), $(\varphi_n-\mathbb{E}[\varphi_n])/\sqrt{\Var[\varphi_n]}$ converges in distribution to $N(0,1)$ as $n\to\infty$,
where $N(0,1)$ denotes a standard normally distributed random variable.
\end{theorem}

We pause to note that one cannot drop either condition (C1) or (C2)
in Theorem~\ref{tetelegy}. It is clear that dropping (C1) is
problematic (for example, set $p_i^{(n)} \in \{0,1\}$ for all $i$ which leads to
a degenerate distribution); as for (C2) the following example shows that we require $\beta$ to
be strictly positive.

\noindent{\bf Example: Condition (C2) cannot be removed}

Consider a tree $\T_{n}$ with $n$ leaves. Leaves $1,\ldots, n-1$
have incident edges that each have length $\frac{1}{\sqrt{n-1}}$ and
all these edges are incident with a vertex that is adjacent to the
root by an edge of length $1$.  Leaf $n$ has edge length $1$ (see
Fig.~\ref{figure2}). Consider a sequence of g-FOB models with
$p_i^{(n)} = s$ for all $i, n$, where $s$ is any number strictly
between $0$ and $1$. Then $\varphi_n =
\frac{1}{\sqrt{n-1}}A_n+B_n+C_n$ where $\frac{1}{\sqrt{n-1}}A_n$ is
the contribution to $\varphi_n$ of the $n-1$ edges that are incident
with leaves $1, \ldots, n-1$; $B_n$ is the contribution to
$\varphi_n$ of the edge that connects these $n-1$ edges to the root
of $\T_{n}$ and $C_n$ is the contribution to $\varphi_n$  of the
edge incident with leaf $n$. Notice that $A_n$ is a sum of $n-1$
i.i.d. binary ($0,1$) random variables, each of which takes the
value $1$ with probability $s$, and  $C_n$ is a binary random
variable which takes the value $1$ with probability $s$.
Consequently, the variance of $\frac{1}{\sqrt{n-1}}A_n$ equals
$s(1-s)$, the same as the variance of $C_n$. Moreover, $B_n$
converges in probability to 1, and $C_n$ is independent of $A_n$ and
$B_n$. Consequently, $\Var[\varphi_n] \rightarrow 2s(1-s)$ as $n
\rightarrow \infty$. Furthermore, by the standard central limit
theorem, $\frac{\frac{1}{\sqrt{n-1}}A_n -
\mathbb{E}[\frac{1}{\sqrt{n-1}}A_n]}{\sqrt{2s(1-s)}}$ converges in
distribution to $N(0, \frac{1}{2})$ (a normal random variable with
mean $0$ and variance $\frac{1}{2}$). Thus,
$(\varphi_n-\mathbb{E}[\varphi_n])/\sqrt{\Var[\varphi_n]}$ converges
to the random variable $N(0,\frac{1}{2}) + W$ where $W$ is
independent of $N(0,\frac{1}{2})$ and takes the value
$\frac{1-s}{\sqrt{2s(1-s)}}$ with probability $s$ and takes the
value $\frac{-s}{\sqrt{2s(1-s)}}$  with probability $1-s$.  In
particular,
$(\varphi_n-\mathbb{E}[\varphi_n])/\sqrt{\Var[\varphi_n]}$ does not
converge in distribution to $N(0,1)$. Notice that in this example,
(C1) is satisfied, but (C2) fails since $\sum_{e \in
E_P(\T_n)}(\lambda_e^{(n)})^2 =2L(n)^2$.

\begin{figure}[ht] \begin{center} \label{figure2}
\resizebox{5.5cm}{!}{
\input{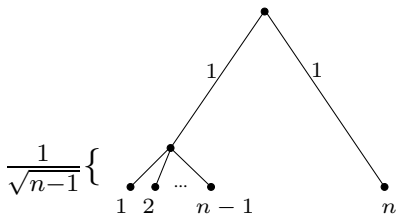}
} \caption{A tree for which future phylogenetic diversity does not
become normally distributed as $n$ grows.}
\end{center}
\end{figure}
\hfill$\Box$

We now provide a brief, informal outline of the approach we use to
prove Theorem~\ref{tetelegy}. The main idea is to decompose $\T_n$
into a `central core' and a large number of `moderately small'
pendant subtrees.  Each edge in the central core separates the root
from enough leaves so that we can be very sure that at least
one of these leaves survives -- consequently the combined
PD-contribution of this central core converges in probability to a
fixed (non-random) function of $n$.  Regarding the pendant subtrees,
their contributions to the PD score are independent and although
they are not identically distributed random variables, their
combined variance grows sufficiently quickly that we can establish a
normal law for their sum by a standard central limit theorem.

\begin{proof}[Proof of Theorem~\ref{tetelegy}]

We first note that it is sufficient to establish Theorem \ref{tetelegy} under (C1) and the seemingly stronger condition:

(C2$^*$) $L(n)=1$, and $\sum_{e \in E_P(\T_n)}(\lambda_e^{(n)})^2
\geq Bn^\beta$ for constants $B>0, \beta >2\alpha$.

To see why, suppose we have established Theorem~\ref{tetelegy} under
(C1), (C2$^*$). For a sequence $\T_n$ (with associated maps
$\lambda^{(n)}$, $p^{(n)}$) satisfying (C1), (C2), let $\mu_e^{(n)}
= L(n)^{-1}\lambda_e^{(n)}$ for each edge $e$ of $\T_n$ and each
$n$. Note that, by Equation~(\ref{fundy}), the normalized $\varphi$
score (namely
$(\varphi_n-\mathbb{E}[\varphi_n])/\sqrt{\Var[\varphi_n]}$)) for
$(\T_n, \mu^{(n)}, p^{(n)})$ equals the normalized $\varphi$ score
for $(\T_n, \lambda^{(n)}, p^{(n)})$  and that $(\T_n, \mu^{(n)},
p^{(n)})$ satisfies (C2$^*$). Thus we will henceforth
assume conditions (C1) and (C2$^*$).

Next, we make a notational simplification: for the remainder of the
proof, we will write $\lambda_e^{(n)}$ as $\lambda_e$ and $p^{(n)}_i$
as $p_i$ (but respecting in the proof that these quantities depend
on $n$). Also, for a sequence of random variables ($Y_n$), we write
$Y_n \xrightarrow{P} a$ to denote that $Y_n$ converges in
probability to a constant $a$, and $Y_n \xrightarrow{D} Y$ to denote
that $Y_n$ converges in distribution to a random variable $Y$.

Since $\beta > 2\alpha$, we may select a value $\gamma$ with
$\alpha < \gamma < \beta/2$, and set $f(n):=n^\gamma$. We partition
the edges of $\T_n$ into two classes $E^n_1$ and $E^n_2$ and we
define a third class $E^n_{12}\subseteq E^n_1$ as follows: Let $n_e$
denote the number of leaves of $\T_n$ that are separated from the
root by $e$. Then set:
\begin{itemize}
\item
$E^n_1$: edges $e$ of $\T_n$ with $n_e\leq f(n)$;
\item
$E^n_2$: edges $e$ of $\T_n$ with $n_e>f(n)$;
\item
$E^n_{12}$: edges $e\in E^n_1$ such that $e$ is adjacent to an edge $f\in E^n_2$.
\end{itemize}

For an edge $e\in E^n_{12}$ of $\T_n$, we make the following
definitions:

\begin{itemize}
\item
$t_e$ denotes the subtree of $\T_n$ consisting of edge $e$ and
all other edges of $\T_n$ that are separated from the root by $e$.
\item
$\varphi_e^n$ denotes the future phylogenetic diversity of
$t_e$,  under the
probabilistic model described above.
\end{itemize}

See Fig. 3  for a
schematic summary of these concepts.

\begin{figure}[ht] \begin{center} \label{figure3}
\resizebox{8cm}{!}{
\input{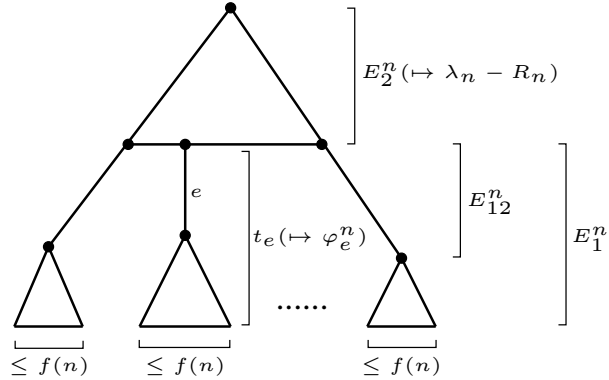}
} \caption{A representation of the decomposition of $\T_n$ in the
proof of Theorem~\ref{tetelegy}.}
\end{center}
\end{figure}

For $\varphi_n$,  Equation (\ref{fundy}) gives
\begin{equation}
\label{deco}
\varphi_n=\sum_{e \in E^n_1} \lambda_eY_e+ \sum_{e \in E^n_2} \lambda_eY_e =
\sum_{e\in E^n_{12}}\varphi_e^n+\sum_{e\in E^n_2}\lambda_eY_e.
\end{equation}
Let $$\lambda_n=\sum_{e\in E^n_2}\lambda_e, Z_n=\sum_{e\in
E^n_{12}}\varphi_e^n, \mbox{  and  } R_n=\sum_{e\in
E^n_2}\lambda_e(1-Y_e).$$ With this notation, we can re-write
(\ref{deco}) as
\begin{center}
\begin{equation}
\varphi_n=\lambda_n+Z_n-R_n.\label{kepletegy}
\end{equation}
\end{center}
\begin{lemma}
\label{lem2}
$R_n \xrightarrow{P} 0$.
\end{lemma}

\begin{proof}
Since $\Var[R_n]=\mathbb{E}[R_n^2]-\mathbb{E}[R_n]^2$ and
$\mathbb{E}[R_n^2]\geq\mathbb{E}[R_n]^2$, it is sufficient to show that
$\mathbb{E}[R_n^2]\to 0$ (the claim that $R_n \xrightarrow{P} 0$ then follows by Chebyshev's inequality). We
have $R_n=\sum_{e\in E^n_2}\lambda_e(1-Y_e)$ and so

\[
R_n^2=\sum_{e,f\in E^n_2}\lambda_e\lambda_f(1-Y_e)(1-Y_f)\leq|E^n_2|\sum_{e\in E^n_2}(1-Y_e),
\]

since $\lambda_e, \lambda_f \leq 1$ by (C2$^*$), and $(1-Y_f) \leq 1$ for all $f \in E^n_2$. Thus,

\begin{equation}
\label{bounonE}
\mathbb{E}[R_n^2]\leq |E^n_2|^2 \cdot \max\{\mathbb{P}[Y_e=0]: e \in E^n_2\}.
\end{equation}
Now, for any edge $e \in E^n_2$ there are at least $n^\gamma -
An^\alpha$ elements $i$ of $C(e)$ for which $p_i \geq \epsilon$ (by
(C1)) and thus $$\mathbb{P}[Y_e=0] \leq (1-\epsilon)^{n^\gamma -
An^\alpha}.$$ Since $|E^n_2|< 2n$, Equation (\ref{bounonE}) and the inequality $\alpha < \gamma$ gives
\[
\mathbb{E}[R_n^2]\leq 4n^2\cdot(1-\epsilon)^{n^\gamma - An^\alpha}
\rightarrow 0 \mbox{ as } n \rightarrow \infty,
\]
as required.
\end{proof}

\begin{lemma}
\label{pep} Under conditions (C1) and (C2$^*$), we have
 $$\sum_{e \in E_P(\T_n)}(\lambda_e^{(n)})^2 P_e(1-P_e)
\geq B\epsilon^2(1+o(1))n^{\beta},$$ where $o(1)$ denotes a term that tends
to $0$ as $n\rightarrow \infty$.
\end{lemma}
\begin{proof}
Let $U_n$ be the set of those pendant edges $e$ of $\T_n$ for which
the leaf incident with $e$ has its survival probability in the
interval $[\epsilon, 1-\epsilon]$, and let $V_n$ denote the set of the
remaining pendant edges of $\T_n$. Clearly,
\begin{equation}
\label{pep1} \sum_{e \in E_P(\T_n)}(\lambda_e^{(n)})^2 P_e(1-P_e)
\geq \epsilon^2 \sum_{e \in U_n}(\lambda_e^{(n)})^2,
\end{equation}
and by (C2$^*$) we have
\begin{equation}
\label{pep2} Bn^\beta \leq \sum_{e \in E_P(\T_n)}(\lambda_e^{(n)})^2
\leq \sum_{e \in U_n}(\lambda_e^{(n)})^2 + |V_n|
\end{equation}
where the last term ($|V_n|$) is an upper bound on $\sum_{e \in
V_n}(\lambda_e^{(n)})^2$ by virtue of the bound $|\lambda_e^{(n)}|
\leq 1$ (by (C2$^*$)).  Since $|V_n| \leq An^\alpha$, Equations
(\ref{pep1}) and (\ref{pep2}) give
$$\sum_{e \in E_P(\T_n)}(\lambda_e^{(n)})^2 P_e(1-P_e) \geq \epsilon^2(Bn^\beta - An^\alpha) = B\epsilon^2(1+o(1))n^{\beta}.$$

\end{proof}

\begin{lemma}\label{lemmaketto}
The random variable $\psi_n=(Z_n-\mathbb{E}[Z_n])/\sqrt{\Var[Z_n]}
\xrightarrow{D} N(0,1)$.
\end{lemma}

\begin{proof}
We can apply a version of the central limit theorem for double
arrays of random variables. The required theorem can be found in
\cite{serfling} and states the following. For each $n$, let
$X_{n1},\ldots,X_{nr}$ be $r=r(n)$ independent random variables with
finite $p$th moments for some $p>2$. Let

\[
A_n=\sum_{j}\mathbb{E}[X_{nj}];\quad B_n=\sum_j\Var[X_{nj}].
\]

If
\begin{equation}\label{cond}
B_n^{-p/2}\sum_j\mathbb{E}[|X_{nj}-\mathbb{E}[X_{nj}]|^p]\to0 \textnormal{ as } n\to\infty,
\end{equation}
then $W_n=(\sum_j X_{nj}-A_n)/\sqrt{B_n}  \xrightarrow{D} N(0,1)$.
We apply this theorem by taking
$\{X_{n1},\ldots,X_{nr}\}=\{\varphi_e^n:e\in E^n_{12}\}$, since the
random variables $\{\varphi_e^n:e\in E^n_{12}\}$ are clearly
independent. With our notation $Z_n=\sum_{e\in
E^n_{12}}\varphi_e^n$, we have $A_n=\mathbb{E}[Z_n]$,
$B_n=\Var[Z_n]$ and $W_n=\psi_n$. Thus, we only need to verify
condition (\ref{cond}) in order to establish Lemma \ref{lemmaketto}.

By Corollary~\ref{lem1}, we have:
\[
\Var[\varphi_e^n]\geq \sum_{f\in E_P(t_e)}\lambda_f^2P_f(1-P_f).
\]
This lower bound and the independence of $\{\varphi_e^n: e \in
E^n_{12}\}$,  implies:
$$B_n=\Var[Z_n]=\sum_{e\in E^n_{12}}\Var[\varphi_e^n] \geq \sum_{e\in E^n_{12}}\sum_{f\in E_P(t_e)}\lambda_f^2P_f(1-P_f)$$
Consequently, by Lemma~\ref{pep}, and the fact that every pendant edge
occurs in $E_P(t_e)$ for some $e \in E_{12}^n$ we obtain,
\begin{equation}
\label{bounding}
B_n\geq B\epsilon^2(1+o(1))n^{\beta}.
\end{equation}
Consider now the absolute central moments in (\ref{cond}). We have
\[
\mathbb{E}[|X_{nj}-\mathbb{E}[X_{nj}]|^p]=\mathbb{E}[|\varphi_e^n-\mathbb{E}[\varphi_e^n]|^p]\leq
L_e^p,
\]
where $L_e$ is the sum of the lengths of the edges of $t_e$. Since
$t_e$ has less than $2n_e$ edges, and the edge lengths are bounded
from above by $1$ (under (C2$^*$)) and $e\in E^n_{12}$ implies
$n_e\leq f(n)$, we obtain $L_e\leq 2n_e\leq 2f(n)$. Now we have

\begin{equation}
\label{bounding2}
\mathbb{E}[|\varphi_e^n-\mathbb{E}[\varphi_e^n]|^p]\leq 2^p f(n)^p.
\end{equation}

Combining the bounds (\ref{bounding}) and (\ref{bounding2}), and noting
that  $|E_{12}^n| \leq 2n$ and $f(n) = n^{\gamma}$ we obtain:

\[
\begin{split}
B_n^{-p/2}\sum_{e\in E^n_{12}}\mathbb{E}[|\varphi_e^n-\mathbb{E}[\varphi_e^n]|^p]&\leq
\frac{|E^n_{12}|2^pf(n)^p}{(B\epsilon^2(1+o(1)))^{p/2}n^{\beta p/2}}\\
&\leq C(p)n^{1+p(\gamma- \beta/2)},
\end{split}
\]
for some constant $C(p)>0$ independent of $n$. Now, since
$\gamma<\beta/2$, the exponent of $n$ in the obtained upper bound is
negative for any $p> (\beta/2- \gamma)^{-1}$. Since there are some
$p>2$ satisfying this inequality and consequently satisfying
condition (\ref{cond}), the proof of Lemma~\ref{lemmaketto} is
complete.
\end{proof}

We return to the proof of Theorem \ref{tetelegy}.  Using Equation
(\ref{kepletegy}) and the notation of Lemma \ref{lemmaketto}, we get

\[
\begin{split}
\frac{\varphi_n-\mathbb{E}[\varphi_n]}{\sqrt{\Var[\varphi_n]}}&=\frac{\lambda_n+Z_n-R_n-(\lambda_n+\mathbb{E}[Z_n]-\mathbb{E}[R_n])}{\sqrt{\Var[\varphi_n]}}\\
&=C_n\psi_n + D_n
\end{split}
\]
where $$C_n =\frac{\sqrt{\Var[Z_n]}}{\sqrt{\Var[\varphi_n]}} \mbox{
and } D_n = -\frac{R_n-\mathbb{E}[R_n]}{\sqrt{\Var[\varphi_n]}}.$$
By Lemma \ref{lem2} and the fact that $\Var[\varphi_n]$ does not
converge to $0$ (by Corollary~\ref{lem1} Lemma~\ref{pep} and
condition (C2$^*$)), we have:
\begin{equation}
\label{sl2} D_n \xrightarrow{P} 0.
\end{equation}
Moreover, by (\ref{kepletegy}),  $\Var[\varphi_n] = \Var[Z_n]+\Var[R_n]-2\Cov[Z_n,R_n]$, so that
$$C_n^{-2} - 1 = \frac{\Var[R_n]}{\Var[Z_n]} - 2\rho \frac{\sqrt{\Var[R_n]}}{\sqrt{\Var[Z_n]}},$$
where $\rho$ is the correlation coefficient of $R_n$ and $Z_n$.
Now, by Lemma~\ref{lem2} we have $\lim_{n \rightarrow \infty} \Var[R_n] = 0$. Thus, since $\Var[Z_n]$ is bounded away from $0$ (by (\ref{bounding})), and $\rho \in [-1,1]$ we have:
\begin{equation}
\label{sl1} \lim_{n \rightarrow \infty} C_n=1.
\end{equation}

To complete the proof of Theorem  \ref{tetelegy}, we apply Slutsky's
Theorem \cite{slutsky} which states that if $X_n, Y_n, W_n$ are
sequences of random variables, and $X_n \xrightarrow{P} a,$ $Y_n
\xrightarrow{P} b,$ (where $a,b$ are constants) and $W_n
\xrightarrow{D} W$ (for some random variable $W$) then $X_nW_n + Y_n
\xrightarrow{D} aW+b$. In our setting, we will take $X_n = C_n, Y_n =
D_n, W_n = \psi_n$, and $W= N(0,1)$ (the standard normal random
variable). The condition that $\psi_n \xrightarrow{D} N(0,1)$ was
established in Lemma~\ref{lemmaketto}, and the conditions $C_n
\xrightarrow{P} 1$, $D_n \xrightarrow{P} 0$ were established in
(\ref{sl1}) and (\ref{sl2}) (note that the convergence of a sequence of
real numbers in (\ref{sl1}) is just a special case of convergence in
probability). Thus,
$$(\varphi_n-\mathbb{E}[\varphi_n])/\sqrt{\Var[\varphi_n]} =
C_n\psi_n+D_n \xrightarrow{D} N(0,1),$$ which completes the proof of
Theorem \ref{tetelegy}.

\end{proof}

\section{Computing the PD distribution}
\label{algorithm}

In this section we describe an algorithm to calculate
the distribution of $\varphi_\cT$ efficiently under the g-FOB model.  An approximate distribution could also be obtained by simulation, but the approach we present here allows us to derive the {\em exact} distribution of
$\varphi_\cT$.  Note that we do not require
conditions (C1) or (C2) in this section.
We make the simplifying assumption that the edge lengths are
non-negative integer-valued, which implies that
$\varphi_\cT$ can only have values in the set
$\{0,1,\ldots,L\}$, where $L=PD(X)=\sum_{e}\lambda_{e}$.  This
assumption is not problematic in practice, as we can rescale all the
edge lengths so that they are (arbitrarily close to) integer
multiples of some small value.

We also assume that the input tree is such that
the root has one outgoing edge and all
other non-leaf vertices have exactly two outgoing edges.
This assumption does not affect the generality of our method, as any tree
can be modified to satisfy it, without changing the distribution for
$\varphi_\cT$: one can resolve multifurcations arbitrarily and
possibly insert an edge below the root, always assigning length 0 to the
newly introduced edges.

Consistent with the notation used before, $\varphi_e$ denotes the
contribution to $\varphi_\cT$ that comes from $e$ and the edges separated
from the root by $e$.
Then, for any edge $e$ and integer $x$, define
\[f_e(x):=\mathbb{P}[\varphi_e=x,\, Y_e=1].\]
Also recall that $P_e=\mathbb{P}[Y_e=1].$

Clearly, if $e$ is the only edge attached to the root of $\cT$, then
$f_e$ and $P_e$ are all that is needed to derive the
distribution of $\varphi_\cT$: simply observe that
\[\mathbb{P}[\varphi_\cT=x] =
\mathbb{P}[\varphi_e=x,\, Y_e=1] + \mathbb{P}[\varphi_e=x,\, Y_e=0] =
 f_e(x) + (1-P_e) \cdot I_{x=0},\]
where $I_{p}$ equals 0 or 1 depending on proposition $p$ being false
or true, respectively.

The algorithm then consists in doing a depth-first (bottom-up)
traversal of all the edges, so that each time an edge $e$ is
visited, the values of $P_e$ and $f_e(x)$, for all $x\in
\{\lambda_e,\lambda_e+1,\ldots,L\}$, are calculated using the
following recursions. We may then use the  $P_e$ and $f_e(x)$ values
of the root edge to calculate the distribution of $\varphi_\cT$.

\subsection*{Recursion for $f_e(x)$.}
\begin{itemize}
\item If $e$ leads into leaf $i$, then
\[f_e(x) \;=\; \mathbb{P}[\varphi_e=\lambda_e,\,Y_e=1]\cdot I_{x=\lambda_e} \;=\;
p_{i} \cdot I_{x=\lambda_e}.\]

\item If $e$ leads into the tail of edges $c$ and $d$, then
\begin{equation}
\label{eq:recursion_rooted}
    f_e(x)  =
    \sum_{i=\lambda_c}^{x-\lambda_e-\lambda_d}
    f_c(i)\cdot f_d(x-\lambda_e-i) +
    (1-P_d)\cdot f_c(x-\lambda_e) +
    (1-P_c)\cdot f_d(x-\lambda_e).
\end{equation}
\end{itemize}

Note that whenever the term $f_c(x-\lambda_e)$ with $x-\lambda_e < \lambda_c$
or the term $f_d(x-\lambda_e)$ with $x-\lambda_e < \lambda_d$
is used in Equation (\ref{eq:recursion_rooted}),
the algorithm will assume that its value is 0
and that therefore there is no need to calculate and store $f_e(x)$
for $x$ outside the range $\{\lambda_e, \lambda_e+1, \ldots , L\}$.

Equation (\ref{eq:recursion_rooted}) is easily proved; we have
\begin{eqnarray*}
f_e(x) & = &
\mathbb{P}[\varphi_e=x,\, Y_c=1,\, Y_d=1] +
\mathbb{P}[\varphi_e=x,\, Y_c=1,\, Y_d=0] +
\mathbb{P}[\varphi_e=x,\, Y_c=0,\, Y_d=1]\nonumber \\
& = &
\mathbb{P}[\varphi_c + \varphi_d = x-\lambda_e,\, Y_c=1,\, Y_d=1]\\
&  & +\, \mathbb{P}[\varphi_c = x-\lambda_e,\, Y_c=1,\, Y_d=0] +
\mathbb{P}[\varphi_d = x-\lambda_e,\, Y_c=0,\, Y_d=1]\nonumber
\end{eqnarray*}
where the second equality is obtained by restating event
$\varphi_e=x$ in terms of $\varphi_c$ and $\varphi_d$, which is
possible once we make assumptions on $Y_c$ and $Y_d$. Thus,
\begin{eqnarray*}
f_e(x) & = & \sum_{i=0}^{x-\lambda_e} \mathbb{P}\left[\varphi_c=i,\,
Y_c=1\right] \cdot \mathbb{P}\left[\varphi_d=x-\lambda_e-i,\,
Y_d=1\right] +
\\
&  &
\mathbb{P}[\varphi_c=x-\lambda_e,\, Y_c=1] \cdot \mathbb{P}[Y_d=0] +
\mathbb{P}[\varphi_d=x-\lambda_e,\, Y_d=1] \cdot \mathbb{P}[Y_c=0]
\nonumber \\
& = &
    \sum_{i=\lambda_c}^{x-\lambda_e-\lambda_d}
    f_c(i)\cdot f_d(x-\lambda_e-i) +
    (1-P_d)\cdot f_c(x-\lambda_e) +
    (1-P_c)\cdot f_d(x-\lambda_e).
\end{eqnarray*}
where the first equality is obtained by using the independence
between the survival events in $C(c)$ and $C(d)$. Note that in the first
expression in the second equality, the range of the sum has been
reduced, as $f_c(i) = 0$ for $i<\lambda_c$ and $f_d(x-\lambda_e-i) =
0$ for $x-\lambda_e-i<\lambda_d$.

\subsection*{Recursion for $P_e$.}
\begin{itemize}
\item If $e$ leads into leaf $i$, then $P_e = p_i$.
\item If $e$ leads into the tail of edges $c$ and $d$, then
$P_e = P_c + P_d - P_cP_d$.
\end{itemize}

\subsection*{Computational complexity}

For any given $e$, the calculation of $P_e$ is
done in $O(1)$ time, whereas that of each of the $f_e(x)$ values
requires $O(x)=O(L)$ time (see recursion (\ref{eq:recursion_rooted})),
giving a total of $O(L^2)$.
Calling $n$ the number of leaves in $\cT$, there are $2n-1$ edges in $\cT$
and the entire procedure takes $O(nL^{2})$ time.

A more efficient version of the algorithm can be obtained by
restricting the calculation of $f_e(x)$ to the values of $x \in
\{\lambda_e, \lambda_e+1, \ldots , L_e\}$, where $L_e$ is the
maximum value that $\varphi_e$ can attain (namely the sum of the
lengths of all the edges separated from the root by $e$, including $e$
itself). Note that the sum in (\ref{eq:recursion_rooted}) can then be further
restricted to the values of $i$ such that $i \leq L_c$ and
$x-\lambda_e-i \leq L_d$.  Using this more efficient algorithm,
it is easy to see that the calculation of all the $f_e(x)$ values
for a given internal edge $e$ takes $O(L_c L_d + L_e)$ time,
where $c$ and $d$ are the edges that $e$ leads into.
Noting that the sum of all the $L_c L_d$ terms, for all sister edges
$c$ and $d$, is bounded above by $L^2$, this shows that the
running time of the entire procedure is $O(L^2 + nL)$.
Since typically every pair of taxa in the tree is separated
by at least one edge of positive length, we have that  $n = O(L)$
and therefore the running time above is equivalent to $O(L^2)$.

Regarding memory requirements, note that each time we calculate the
information relative to $e$ (namely $P_e$ and $f_e(x)$), the
information relative to the edges it leads to (if any) can be
deleted, as it will never be used again. So, at any given moment the
information of at most $n$ `active' edges needs to be stored. If we
use the range restriction just described, the sizes of the $f_e(x)$
vectors for all the active edges sum to a number bounded above by
$n+L$, and therefore the algorithm requires $O(n+L)$ space,
equivalent to $O(L)$ if $n = O(L)$.

\section{Extension to unrooted PD}
\label{unroot}

There is a simple modification of the definition of
phylogenetic diversity that is also relevant in biology
(\cite{fai06a}, \cite{pardi}). Given a subset $X'$ of $X$, we can
evaluate the sum of the lengths of the edges in the minimum subtree
connecting (only) the leaves in $X'$. This score -- which we will
denote by $uPD(X')$ and refer to as the `unrooted PD' score of $X'$
-- is equivalent to $PD(X')$ if the path connecting two leaves in
$X'$ traverses the root of $\T$. However, in general, $uPD(X') \leq
PD(X')$ (Fig.~\ref{figure1b} shows an example where $uPD(X') <
PD(X')$).
\begin{figure}[ht] \begin{center} \label{figure1b}
\resizebox{13cm}{!}{
\input{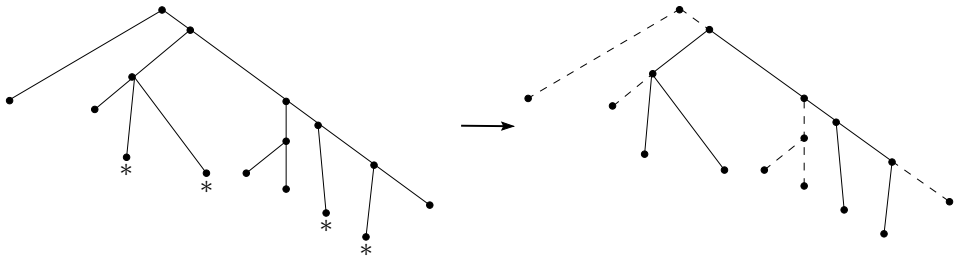}
} \caption{If only the taxa marked * in the tree on the left survive
then future `unrooted' phylogenetic diversity is the sum of the
lengths of the solid edges in the tree on the right. Notice that the
$uPD$ value in this example is less than the $PD$ value ({\em c.f.}
Fig. 1).}
\end{center}
\end{figure}
This alternative concept of phylogenetic diversity has the advantage
that it can be defined on either rooted or unrooted phylogenetic
trees. Of course, the g-FOB model is also defined naturally on
unrooted trees, and so it makes sense to consider the distribution of $uPD$ under the g-FOB model in this more general setting. A natural question is whether Theorem~\ref{tetelegy} is still valid,
(that is, is the future uPD of (rooted or unrooted) trees also
asymptotically normal under conditions (C1) and (C2)?). We now answer
this question (affirmatively) and also show how  to extend
the computation of the exact future PD distribution to unrooted
trees.

Let the random variable $\varphi'=\varphi'_\T$ denote the uPD score
of the random subset $X'$ of $X$ (consisting of those taxa that will
still exist at some time $t$ in the future). We call $\varphi'$ the
{\em future unrooted phylogenetic diversity}.  In this model, we have
\begin{equation}
\label{fundynoroot}
\varphi' = \sum_e \lambda_eY'_e,
\end{equation}
where $Y'_e$ is the binary random variable which takes the value $1$
if $e$ lies on at least one path between some pair of taxa in $X'$, and
which is $0$ otherwise. Moreover,
\begin{equation}
\label{fundy2noroot}
\mathbb{P}[Y'_e = 1] = (1 - \prod_{i \in X_1(e)}(1-p_i))(1 - \prod_{j \in X_2(e)}(1-p_j)),
\end{equation}
where $X_1(e)$ and $X_2(e)$ are the bipartition of $X$ consisting of
the two subsets of $X$ that are separated by edge $e$. Thus if we
let $P_i(e)$ denote the probability that at least one taxon in
$X_i(e)$ survives (for $i\in\{1,2\}$), then the  expected value of
$\varphi'$ (analogous to (\ref{fundy3})) is
\begin{equation}
\label{fundy3noroot} \mathbb{E}[\varphi'] =\sum_e\lambda_e P_1(e) P_2(e).
\end{equation}
Regarding $\Var[\varphi']$ there is an analogous formula to that given in
Lemma~\ref{fundy4}.

Consider now a sequence $\T_1,\T_2,\ldots,\T_n,\ldots$ of (rooted or
unrooted) phylogenetic trees where $\T_n$ has $n$ leaves, and assume
that this sequence satisfies conditions (C1) and (C2) when each
$\T_n$ has associated edge length and leaf survival probability
functions. It can be shown that Theorem~\ref{tetelegy} is still
valid for uPD; that is, under the same conditions,
$(\varphi'_n-\mathbb{E}[\varphi'_n])/\sqrt{\Var[\varphi'_n]}$
converges in distribution to $N(0,1)$ as $n\to\infty$.

To establish this asymptotic normality of $\varphi'_n$ under conditions (C1) and (C2$^*$) (and thereby (C1) and (C2)) requires slight modifications to the proof of Theorem~\ref{tetelegy}, and we now provide an outline of the argument. The main difference is that now each edge $e$ induces a bipartition $X=X_1(e)\cup X_2(e)$ of the taxon set and so we decompose $\T_n$ in a slightly different way. For simplicity, assume that $|X_1(e)|\leq |X_2(e)|$ and consider the following edge sets (the definition of the function $f(n)$ is as in the rooted case):\\

\noindent $E^n_1$: edges $e$ of $\T_n$ with $|X_1(e)|\leq f(n)$.\\
$E^n_2$: edges $e$ of $\T_n$ with $|X_1(e)|>f(n)$.\\
$E^n_{12}$: edges $e\in E^n_1$ such that $e$ is adjacent to an edge $f\in E^n_2$.\\

For $\varphi_n'$ we obtain the following equation:

\begin{equation}
\label{fundy5noroot}
\varphi_n' = \sum_{e\in E_1^n}\lambda_e Y_e'+ \sum_{e\in E_2^n}\lambda_e Y_e'=\sum_{e\in E_1^n}\lambda_e Y_e'+\lambda_n-R_n',
\end{equation}
where $\lambda_n=\sum_{e\in E_2^n}\lambda_e$ and $R_n'=\sum_{e\in
E_2^n}\lambda_e(1-Y_e')$. For an edge $e\in E_{12}^n$, let $t_e$
denote the subtree with root edge $e$ and with leaf set $X_1(e)$.
Let $(\varphi_e^n)'$ denote the contribution to $\varphi_n'$ by the edges
in $t_e$. Furthermore, let $\varphi_e^n$ be the rooted future
phylogenetic diversity of $t_e$, $Z_n=\sum_{e\in
E_{12}^n}\varphi_e^n$ as in the rooted case,
$W_e=\varphi_e^n-(\varphi_e^n)'$ and $V_n=\sum_{e\in E_{12}^n}W_e$. With
this notation, we get
\begin{equation}
\varphi_n' = \sum_{e\in E_{12}^n}(\varphi_e^n)'+\lambda_n-R_n'=\sum_{e\in E_{12}^n}\varphi_e^n-\sum_{e\in E_{12}^n}W_e+\lambda_n-R_n'=Z_n-V_n+\lambda_n-R_n'.
\end{equation}

Now we can apply Lemma~\ref{lemmaketto} and Slutsky's Theorem to complete the proof.

\subsection{Computing the uPD distribution}
Finally we show how the algorithm described in Section~\ref{algorithm} for
computing the PD distribution can be modified to calculate the distribution of
unrooted PD.
As before, we assume the edge lengths are non-negative integers and we
preprocess $\cT$ (possibly rooting it in an arbitrary vertex) so that the
number of outgoing edges is 1 for the root and 2 for all the other non-leaf
vertices. Since $\cT$ is now rooted, $C(e)$ and the random variables $Y_e$
are well defined.
We also define $\varphi'_e$ as the uPD of the surviving taxa in $C(e)$.
Then, for any integer $x$, define\[f'_e(x):=\mathbb{P}[\varphi'_e=x,\,Y_e=1].\]

As before, if $e$ is the root edge of $\cT$, then
$f'_e$ and $P_e$ are sufficient to derive the
distribution of $\varphi'_\cT$:
\[\mathbb{P}[\varphi'_\cT=x]=f'_e(x)+(1-P_e) \cdot I_{x=0}.\]

An algorithm to calculate the distribution of $\varphi'_\cT$ can be obtained
with a simple modification of the algorithm for $\varphi_\cT$:
for each edge $e$, in addition to calculating $P_e$ and $f_e(x)$, also
calculate $f'_e(x)$, for all $x\in \{0,1,\ldots,L\}$.
For this purpose, the following recursion is used
(note that the $f'_e$ values may depend on $f_c$ and $f_d$ as well as on
$f'_c$ and $f'_d$, which is why we retain the calculation of the $f_e$
values even though they are not directly implicated in determining
$\mathbb{P}[\varphi'_\cT=x]$).

\subsection*{Recursion for $f'_e(x)$.}

\begin{itemize}
\item If $e$ leads into leaf $i$, then
\[f'_e(x) \;=\; p_i \cdot I_{x=0}.\]

\item If $e$ leads into the tail of edges $c$ and $d$, then
\begin{equation}
f'_e(x) =
\sum_{i=\lambda_c}^{x-\lambda_d}
f_c(i) \cdot f_d(x-i) +
(1-P_d) \cdot f'_c(x) +
(1-P_c) \cdot f'_d(x),
\end{equation}
\end{itemize}

which is proved in a way similar to (\ref{eq:recursion_rooted}):
\begin{eqnarray*}
f'_e(x)
& = &
\mathbb{P}[\varphi'_e=x,\, Y_c=1,\, Y_d=1] +
\mathbb{P}[\varphi'_e=x,\, Y_c=1,\, Y_d=0] +
\mathbb{P}[\varphi'_e=x,\, Y_c=0,\, Y_d=1]\\
& = &
\mathbb{P}[\varphi_c+\varphi_d=x,\, Y_c=1,\, Y_d=1] +
\mathbb{P}[\varphi'_c=x,\, Y_c=1,\, Y_d=0] +
\mathbb{P}[\varphi'_d=x,\, Y_c=0,\, Y_d=1]\\
& = &
\sum_{i=\lambda_c}^{x-\lambda_d}
f_c(i) \cdot f_d(x-i) +
(1-P_d) \cdot f'_c(x) +
(1-P_c) \cdot f'_d(x).
\end{eqnarray*}

\section{Concluding remarks}

The main result of this paper (Theorem~\ref{tetelegy}) has been to
establish a limiting normal distribution for future PD on large
phylogenetic trees.  This theorem assumes an underlying generalized
`field of bullets' model, and imposes two further mild conditions
(conditions (C1) and (C2)).  In this setting Theorem~\ref{tetelegy}
reduces the problem of computing the distribution of future PD to
that of determining just two parameters -- its mean and variance --
and these can be readily computed by Equation (\ref{fundy3}) and
Lemma~\ref{fundy4}. Using the resulting normal distribution one can
easily compute the probability under the model that future PD will
fall below any given critical value. This may also be helpful in
designing strategies to minimize this probability (analogous to the
`Noah's Ark problem' which tries to maximize expected future PD).

In practice, the use of a normal distribution based on
Theorem~\ref{tetelegy} requires that the number of taxa is moderate
($>50$), that the survival probabilities are not too extreme
(condition (C1)), and that the length of the pendant edges on
average are not too small in relation to the largest edge length in
the tree (condition (C2)). If these conditions are violated, it would
be prudent to use the exact algorithm we have described in the
paper, as this requires neither a large number of taxa nor condition
(C1) or (C2). To apply this algorithm may involve some small
adjustment to the edge lengths to make them integral multiples of
some common value.

Regarding the `mild conditions' for Theorem~\ref{tetelegy} (namely
(C1) and (C2)), we showed that neither can be dropped completely from
the statement of the theorem. However it is likely that both
conditions could be weakened somewhat, though at the risk of
complicating the description of the conditions and the proof of
Theorem \ref{tetelegy}.

It would be interesting to explore other extinction models that
weaken the strong assumption in the g-FOB model that taxon
extinction events are independent. One such model would regard
extinction as a continuous-time Markov process in which the
extinction rate of a taxon $i$ at any given time $t$ is the product
of an intrinsic extinction rate $r_i$, with a factor that depends on
the set of species in the tree that are extant at time $t$.  In
general, such processes could be very complex, so a first step would
be to identify a simple model that nevertheless captures more
biological realism than the g-FOB model.

\section{Appendix}

\noindent{\em Condition (C2) is satisfied in expectation for
trees generated under a continuous-time pure-birth model.}

Consider a model where each lineages persist for a random period of
time before speciating, and that these persistence times are i.i.d.
random variables with exponential distribution with mean $s>0$ (This
model is sometimes called the `Yule model' in phylogenetics).  Now
suppose we sample the process at some time during which the tree has
$n$ leaves. Let $\T_n$ denote this tree (with its associated edge
lengths) and let $\mu_P(n)$ denote the average length of the pendant
edges of $\T_n$.

\begin{proposition}
Under a constant birth model with speciation rate $s$ and $\beta \in
(0,1)$ there is a constant $B>0$ for which the expected value of
$$\sum_{e \in E_P(\T_n)}\left(\lambda_e^{(n)}\right)^2 - B
n^{\beta}L(n)^2$$ is strictly positive for all $n \geq 3$.
\end{proposition}

\begin{proof}
By using the inequality $\sum_{i=1}^n x_i^2 \geq
\frac{1}{n}(\sum_{i=1}^n x_i)^2$ we have $${\mathbb E}[\sum_{e \in
E_P(\T_n)}\left(\lambda_e^{(n)}\right)^2 - B n^{\beta}L(n)^2] \geq n
{\mathbb E}[\mu_P(n)]^2 - Bn^{\beta}{\mathbb E}[L(n)^2]$$ and the
proposition follows (by choice of a sufficiently small value of
$B>0$) once we establish the following two results.
\begin{itemize}
\item[(i)] ${\mathbb E}[\mu_P(n)] \geq s/6$ for all $n \geq 3$.
\item[(ii)] $n^{-\eta}{\mathbb E}[L(n)^2] \rightarrow 0$ as $n \rightarrow
\infty$, for any $\eta>0$.
\end{itemize}

To establish result (i), let $S_n$ denote the sum of the lengths of
the pendant edges of $\T_n$ up till the point when the number of
species first changes from $n-1$ to $n$, and excluding the (length
of the) pendant edge on which this speciation event occurs. Thus
$S_n$ is a sum of lengths of $n-2$ pendant edges.  [For example,
$S_3$ has an exponential distribution with mean $s/2$, as it is the
length of the edge that does not first speciate, up until the time
when one of the two edges in the tree first speciates]. Since we are
observing the tree $\T_n$ at some later time (but while it still has
$n$ leaves) then we clearly have:
\begin{equation}
\label{muP} \mu_P(n) \geq \frac{1}{n}{S_n}.
\end{equation}
We will derive a recursion for the sequence $({\mathbb E}[S_n], n =
3,4, \ldots)$.  Let $\theta_n$ be an exponentially-distributed
random variable with mean $s/n$.  Now, the random variable $S_{n+1}$
takes the value $S_n + (n-1)\theta_n$, with probability $2/n$ (this
is the case where the next speciation event occurs on one of the two
edges that develop from the last speciation event). Otherwise (and
so with probability $1-2/n$),  $S_{n+1}$ takes the value $S_n +
(n-1)\theta_n - \lambda_e$, where $\lambda_e$ is the length of one
of the $n-2$ pendant edges that contribute to $S_n$ (selected
uniformly at random from this set of edges).

Consequently,
\begin{eqnarray*}
{\mathbb E}[S_{n+1}] & = & \frac{2}{n}({\mathbb E}[S_{n}] +
 (n-1)\frac{s}{n}) + (1-\frac{2}{n})({\mathbb
 E}[S_{n}](1-\frac{1}{n-2}) +(n-1)\frac{s}{n}) \\
& = & \frac{n-1}{n}({\mathbb E}[S_n]+s).
\end{eqnarray*}
By using the initial condition ${\mathbb E}[S_3]=s/2$, and this
recursion, we have that ${\mathbb E}[S_n] = ns/2-s$ for all $n \geq
3$. Taking expectations on both sides of inequality~(\ref{muP})
gives
$$
{\mathbb E}[\mu_P(n)] \geq \frac{1}{n}(ns/2-s) = \frac{s}{2} -
\frac{s}{n} \geq \frac{s}{6},
$$
for all $n \geq 3$, thus proving (i).

To establish result (ii), observe that length of the longest edge in
$\T_n$ (namely $L(n)$) is bounded above by the length of the longest
edge in the tree obtained from $\T_n$ by allowing each leaf to
evolve until it next speciates. Now, the lengths of the edges in
this resulting trees is a set of $|E(\T_n)|$ independent random
variables each having an exponential distribution with mean $s$
(here $|E(\T_n)|$ is the number of edges of $\T_n$, which is at most
$2n-1$). Thus, if we let $Y_n$ be the maximum of $2n-1$ i.i.d
exponentially-distributed random variables, each with mean $s$, then
$L(n) \leq Y_n$. Moreover, for any $x>0$ we have: \begin{equation}
\label{ineq1} {\mathbb E}[Y_n^2] =  \int_{0}^{\infty}{\mathbb P}[Y_n^2
> y] dy \leq x^2 + \int_{x^2}^{\infty}{\mathbb P}[Y_n^2 > y] dy,
\end{equation} where the first equality in (\ref{ineq1}) is a standard identity in
probability theory for any non-negative random variable $Y_n^2$.
Now, by Boole's inequality, $${\mathbb P}[Y_n^2> y] ={\mathbb P}[Y_n> \sqrt{y}]  \leq (2n-1)
\exp(-\sqrt{y}/s).$$ Making the substitution $y = x^2+t^2$,
and applying the inequality $\sqrt{x^2+t^2} \geq \frac{x+t}{\sqrt{2}}$ we
obtain $$\int_{x^2}^{\infty}{\mathbb P}[Y_n^2 > y] dy \leq
2(2n-1)\exp(-\frac{x}{s\sqrt{2}}) \int_{t=0}^\infty t\exp(-\frac{t}{s\sqrt{2}})
dt.$$ Thus, taking $x= cs\log(n)$, for any $c> \sqrt{2}$, in (\ref{ineq1}) we obtain
$${\mathbb E}[L(n)^2] \leq {\mathbb E}[Y_n^2] \leq c^2s^2(\log(n))^2 + o(1),$$
(where $o(1)$ is a term that tends to $0$ as $n \rightarrow \infty$)
and result (ii) now follows.
\end{proof}


\begin{thebibliography}{99}




\addcontentsline{toc}{chapter}{References}



\bibitem{slutsky} Durrett, R., 1991.  Probability: Theory and Examples,
Wadsworth and Brooks/Cole, Belmont, California.

\bibitem{yule}   Edwards, A.W.F., 1970. Estimation of the branch points of a branching diffusion
process. (With discussion.). J. Roy. Statist. Soc. Ser. B.,
32, 155--174.

\bibitem{fai92} Faith, D.P., 1992. Conservation evaluation and phylogenetic
diversity. Biol. Conserv. 61, 1–-10.

\bibitem{fai06} Faith, D.P., 2006. The role of the phylogenetic diversity measure, PD, in
bio-informatics: Getting the definition right. Evol. Bioinf. Online.


\bibitem{fai06a} Faith, D.P., Baker, A.M., 2006. Phylogenetic diversity (PD) and biodiversity conservation: some bioinformatics challenges. Evol. Bioinf. Online.

\bibitem{hart06}  Hartmann, K., Steel, M., 2006.
Maximimizing phylogenetic diversity in biodiverstity conservation:
greedy solutions to the Noah's Ark problem. Syst. Biol.
55(4), 644--651.

\bibitem{hart} Hartmann, K., Steel, M., 2007. Phylogenetic diversity: From combinatorics to ecology.
In: Reconstructing evolution: New mathematical and computational
approaches (eds. O. Gascuel and M. Steel), Oxford University Press, Oxford, UK.

\bibitem{moo} Mooers, A.O., Heard, S.B., Chrostowski, E., 2005.
Evolutionary heritage as a metric for conservation. Pages 120–-138
in Phylogeny and conservation (A. Purvis, T. Brooks, and J.
Gittleman, eds.). Cambridge University Press, Cambridge, UK.

\bibitem{nee} Nee, S., May, R. M.,  1997. Extinction and the
loss of evolutionary history. Science, 278(5338), 692–-694.

\bibitem{pardi}
Pardi, F., Goldman, N., 2007. Resource-aware taxon selection for
maximizing phylogenetic diversity, Syst. Biol. 56(3), 431–-444.


\bibitem{pur}  Purvis, A., Agapow, P-M., Gittleman, J.L., Mace, G.M., 2000. Nonrandom extinction and the loss of
evolutionary history. Science 288,  328--330.


\bibitem{raup} Raup, D.M.,  1993.  Extinction: bad genes or bad luck? Oxford Univ.
Press, Oxford.

\bibitem{serfling} Serfling, R.J., 1980. Approximation theorems of mathematical statistics, Wiley, New York.


\bibitem{vaz} Vazquez, D.P., Gittleman, J.L., 1998.
Biodiversity conservation: Does phylogeny matter? Current Biol.
8, 379--381.


\bibitem{noah} Weitzman, M.L., 1998. The Noah's ark problem. Econometica, 66(6), 1279--1298.

\end{thebibliography}
\end{document}